\documentclass[11pt,a4paper,reqno]{amsproc}

\usepackage{datetime}
\usepackage[pdfborder={0 0 0}]{hyperref}
\usepackage{upref,cases}
\hypersetup{citebordercolor=.1 .6 1}

\IfFileExists{mymtpro2.sty}{\usepackage[subscriptcorrection]{mymtpro2}
	\def\one{\mathbb{1}}
  \def\cup{\cupprod}
  
  \def\bigcup{\bigcupprod}
  
  \def\bigcupdisjoint{\mathop{\kern10pt\raisebox{4pt}{$\cdot$}\kern-12pt\bigcup}\limits}
}%
	{\usepackage{amssymb,bm,dsfont}
	 \let\mathbold\bm
	 	\def\one{\mathds{1}}%
}

\numberwithin{equation}{section}
\newtheoremstyle{ttheorem}%
       {1.8ex\@plus1ex}                
       {2.1ex\@plus1ex\@minus.5ex}      
       {\itshape}           
       {0pt}                   
       {\bfseries}          
       {.}                  
       {.5em}               
       {}                

\newtheoremstyle{ddefinition}%
       {1.8ex\@plus1ex}                
       {2.1ex\@plus1ex\@minus.5ex}      
       {}           
       {0pt}                   
       {\bfseries}           
       {.}                  
       {.5em}               
       {}                

\newtheoremstyle{rremark}%
       {1.8ex\@plus1ex}                
       {2.1ex\@plus1ex\@minus.5ex}      
       {\normalfont}        
       {0pt}                   
       {\bfseries}           
       {.}                  
       {.5em}               
       {}                   

\theoremstyle{ttheorem}
\newtheorem{theorem}{Theorem}[section]
\newtheorem{lemma}[theorem]{Lemma}
\newtheorem{proposition}[theorem]{Proposition}
\newtheorem{corollary}[theorem]{Corollary}

\theoremstyle{ddefinition}
\newtheorem{definition}[theorem]{Definition}

\theoremstyle{rremark}
\newtheorem{remark}[theorem]{Remark}
\newtheorem{myremarks}[theorem]{Remarks}
\newtheorem{myexamples}[theorem]{Examples}

\newenvironment{remarks}{\begin{myremarks}\begin{nummer}}%
    {\end{nummer}\end{myremarks}}
    {\end{nummer}\end{myexamples}}

\newcounter{numcount}
\newcommand{\labelnummer}{(\roman{numcount})}%

\makeatletter
\providecommand{\showkeyslabelformat}[1]{\relax}        
\let\mysaveformat\showkeyslabelformat                   %
\def\myformat#1{\raisebox{-1.5ex}{\mysaveformat{#1}}}   %

\newenvironment{nummer}%
  {\let\curlabelspeicher\@currentlabel%
    \begin{list}{\textup{\labelnummer}}%
      {\usecounter{numcount}\leftmargin0pt%
        \topsep0.5ex\partopsep2ex\parsep0pt\itemsep0ex\@plus1\p@%
        \labelwidth2.5em\itemindent3.5em\labelsep1em%
      }%
    \let\saveitem\item%
    \def\item{\saveitem%
      \def\@currentlabel{\curlabelspeicher\kern.1em\labelnummer}}%
    \let\savelabel\label%
    \def\label##1{{\ifnum\thenumcount=1\let\showkeyslabelformat\myformat\fi\savelabel{##1}}%
										{\def\@currentlabel{\labelnummer}%
									 	\let\showkeyslabelformat\@gobble
									 	\savelabel{##1item}%
										}%
	   							}%
  }{\end{list}}%

  {\let\curlabelspeicher\@currentlabel%
    \begin{list}{\textup{\labelnummer}}%
      {\usecounter{numcount}\leftmargin0pt%
        \topsep0.5ex\partopsep2ex\parsep0pt\itemsep0ex\@plus1\p@%
        \labelwidth2.5em\itemindent0em\labelsep1em%
        \leftmargin2.5em}%
    \let\saveitem\item%
    \def\item{\saveitem%
      \def\@currentlabel{\curlabelspeicher\kern.1em\labelnummer}}%
    \let\savelabel\label%
    \def\label##1{{\ifnum\thenumcount=1\let\showkeyslabelformat\myformat\fi\savelabel{##1}}%
										{\def\@currentlabel{\labelnummer}%
									 	\let\showkeyslabelformat\@gobble
									 	\savelabel{##1item}%
										}%
    							}%
  }{\end{list}}%

\def\itemref#1{\ref{#1item}}

\def\section{\@startsection{section}{1}%
  \z@{1.3\linespacing\@plus\linespacing}{.5\linespacing}%
  {\normalfont\bfseries\centering}}

\def\subsection{\@startsection{subsection}{2}%
  \z@{.8\linespacing\@plus.5\linespacing}{-1em}%
  {\normalfont\bfseries}}

\def\nlsubsection{\@startsection{subsection}{2}%
  \z@{.8\linespacing\@plus.5\linespacing}{.1ex}%
  {\normalfont\bfseries}}

\let\@afterindenttrue\@afterindentfalse%

\renewenvironment{proof}[1][\proofname]{\par \normalfont
  \topsep6\p@\@plus6\p@ \trivlist 
  \item[\hskip\labelsep\scshape
    #1\@addpunct{.}]\ignorespaces
}{%
  \qed\endtrivlist
}

\def\ps@firstpage{\ps@plain
  \def\@oddfoot{\normalfont\scriptsize \hfil\thepage\hfil
     \global\topskip\normaltopskip}%
  \let\@evenfoot\@oddfoot
  \def\@oddhead{
    \begin{minipage}{\textwidth}
      \normalfont\scriptsize
      \emph{\insertfirsthead}
    \end{minipage}}
  \let\@evenhead\@oddhead 
}

\def\insertfirsthead{}

\def\@cite#1#2{{%
 \m@th\upshape\mdseries[{#1}{\if@tempswa, #2\fi}]}}
\makeatother

\newcommand{\I}{\mathcal{I}}
\newcommand{\F}{\mathcal{F}}
\newcommand{\Nn}{\mathcal{N}}

\newcommand{\SO}{\mathcal{S}}

\newcommand{\C}{\mathbb{C}}
\newcommand{\N}{\mathbb{N}}

\newcommand{\R}{\mathbb{R}}

\newcommand{\Z}{\mathbb{Z}}

\renewcommand{\le}{\leqslant}
\renewcommand{\ge}{\geqslant}
\renewcommand{\i}{\mathrm{i}}
\newcommand{\e}{\mathrm{e}}

\DeclareMathOperator{\tr}{tr}

\DeclareMathOperator{\supp}{supp}
\DeclareMathOperator{\dist}{dist}

\providecommand{\mathbold}[1]{\mathbf{#1}}
\providecommand{\wtilde}[1]{\widetilde{#1}}
\providecommand{\what}[1]{\widehat{#1}}

\providecommand{\bigcupdisjoint}{\mathop{\kern7pt\raisebox{6pt}{$\cdot$}\kern-9.5pt\bigcup}\limits}

\let\<\langle
\let\>\rangle

\providecommand{\abs}[1]{\lvert#1\rvert}
\providecommand{\norm}[1]{\lVert#1\rVert}

\providecommand{\bigabs}[1]{\bigl\lvert#1\bigr\rvert}
\providecommand{\Bigabs}[1]{\Bigl\lvert#1\Bigr\rvert}
\providecommand{\biggabs}[1]{\biggl\lvert#1\biggr\rvert}

\providecommand{\bignorm}[1]{\bigl\lVert#1\bigr\rVert}
\providecommand{\Bignorm}[1]{\Bigl\lVert#1\Bigr\rVert}

\providecommand{\bigparens}[1]{\bigl(#1\bigr)}
\providecommand{\Bigparens}[1]{\Bigl(#1\Bigr)}
\providecommand{\biggparens}[1]{\biggl(#1\biggr)}

\providecommand{\Bigbraces}[1]{\Bigl\{#1\Bigr\}}

\newcommand{\ac}{\mathrm{ac}}
\newcommand{\loc}{\mathrm{loc}}
\newcommand{\HS}{\mathrm{HS}}

\providecommand{\st}{:\,}

\newcommand{\upto}{\uparrow}

\newcommand{\Oh}{\mathrm{O}}
\newcommand{\oh}{\mathrm{o}}

\newcommand{\dotid}{\,\boldsymbol\cdot\,}           

\DeclareMathOperator{\Borel}{Borel}      

\newcommand{\1}{1}

\newcommand{\upd}{\mathrm{d}}
\renewcommand{\d}{\upd}   
\newcommand{\dx}{\d x}
\newcommand{\dt}{\d t}
\newcommand{\dy}{\d y}
\newcommand{\dz}{\d z}

\newcommand{\spec}{\sigma}               

\begin{document}


\title{Anderson's orthogonality catastrophe}

\author[M.\ Gebert]{Martin Gebert}
\author[H.\ K\"uttler]{Heinrich K\"uttler}
\author[P.\ M\"uller]{Peter M\"uller}

\address{Mathematisches Institut,
  Ludwig-Maximilians-Universit\"at M\"unchen,
  Theresienstra\ss{e} 39,
  80333 M\"unchen, Germany}

\email{gebert@math.lmu.de}
\email{kuettler@math.lmu.de}
\email{mueller@lmu.de}

\thanks{Work supported by Sfb/Tr 12 of the German Research Council
(Dfg). The final publication will appear in
\href{http://dx.doi.org/10.1007/s00220-014-1914-3}{Commun.\ Math.\ Phys.}
and is available at link.springer.com.}

\begin{abstract}
  We give an upper bound on the modulus of the ground-state overlap of two non-interacting fermionic
  quantum systems with $N$ particles in a large but finite volume $L^{d}$ of $d$-dimensional Euclidean space.
  The underlying one-particle Hamiltonians of the two systems are standard Schr\"odinger
  operators that differ by a non-negative compactly supported scalar potential.
  In the thermodynamic limit, the bound exhibits an asymptotic power-law decay in the system size $L$, showing that the ground-state overlap vanishes for macroscopic systems.
  The decay exponent can be interpreted in terms of the total scattering cross section
  averaged over all incident directions. The result confirms and generalises
  P.~W.\ Anderson's informal computation [Phys.\ Rev.\ Lett.\ \textbf{18}, 1049--1051, 1967].
\end{abstract}

\maketitle

%
\section{Introduction}

Anderson's orthogonality catastrophe (AOC) is an intrinsic effect in
many-body fermionic systems. It arises when a system reacts to a sudden perturbation. For instance, one may think of a sudden X-ray excitation of a core electron in an atom, leaving behind a hole in a core shell. In bulk metals the AOC manifests itself in the asymptotic vanishing
\begin{equation}
	\label{catastrophe}
 	\boldsymbol\langle \Phi^N_{L},\Psi^N_{L}\boldsymbol\rangle \sim L^{-\gamma/2}
\end{equation}
of the overlap of the $N$-body ground states $\Phi^N_{L}$ and $\Psi^N_{L}$ of a given
fermionic system in a box of length $L$ with and without a perturbation
in the thermodynamic limit $L\to\infty$, $N\to\infty$, $N/L^{d} \to \text{const.} >0$. Here, $d\in\N$ is the spatial dimension.

The AOC leads to physically important effects, like Fermi-edge
singularities in the X-ray edge problem \cite{PhysRev.178.1097,RevModPhys.62.929}. It has proved to be an extremely robust
phenomenon with consequences reaching far beyond single-impurity
problems, and it continues to attract attention in the physics
literature. Recent studies
\cite{PhysRevB.72.125301,PhysRevLett.106.107402,PhysRevLett.108.190601,PhysRevB.85.155413}
considered this effect in optical absorption or emission involving a single quantum-dot
level hybridising with a Fermi sea. The absorption or emission of a photon induces a
sudden local perturbation in the Fermi sea with consequences similar to that in the
classical X-ray edge problem. Other manifestations in mesoscopic systems of current interest, such as graphene, can be found, e.g., in \cite{PhysRevB.72.035310,PhysRevB.76.115407,PhysRevB.82.125312}.

P.\ W.\ Anderson was the first to explain the behaviour
\eqref{catastrophe} in the late 1960ies. He considered a non-interacting Fermi gas in three dimensions and its perturbation by a compactly supported single-particle potential.
In this setting he used Hadamard's inequality to estimate the Slater determinant of the two ground states from above -- see also Lemma~\ref{lemmaoverlap} -- and a subsequent informal computation, which led to
\begin{equation}
	\label{catas-bound}
	\boldsymbol\langle \Phi^N_{L},\Psi^N_{L}\boldsymbol\rangle = \Oh(L^{-\gamma/2})
\end{equation}
in the thermodynamic limit \cite{PhysRevLett.18.1049}. Furthermore, Anderson expressed the decay exponent $\gamma$ in terms of the (single-particle) scattering phases associated with the perturbation.
Later on in the same year, he found a way \cite{PhysRev.164.352} to
circumvent Hadamard's inequality and arrived at the asymptotics
\eqref{catastrophe} with an exponent $\gamma$ bigger than that in \eqref{catas-bound}. After some controversies about the correctness of interchanging limits, the asymptotics of \cite{PhysRev.164.352} was confirmed in an adiabatic approach \cite{RiSi71,Ham71,KaYo78}.
Both this asymptotics and the bound \eqref{catas-bound} from \cite{PhysRevLett.18.1049} are
now taken for granted in the physics literature as a fundamental property of Fermi gases.

Only very little is known about AOC from a rigorous mathematical point of view.
The adiabatic approach to AOC was revisited by Otte \cite{1087.47015},
who rigorously derives a limit expression for the overlap in terms of
the solution of a Wiener-Hopf equation, thereby clarifying a
discussion on the correctness of limits in \cite{RiSi71, Ham71}, see also
\cite{1045.81584} for related work. Unfortunately, this does not
allow  the thermodynamic limit to be controlled, which would be necessary for proving
\eqref{catastrophe}.
Likewise, the upper bound \eqref{catas-bound} awaits a sound
mathematical treatment. This is the goal of the present paper. We will
prove \eqref{catas-bound} with the same decay exponent $\gamma$ as in
\cite{PhysRevLett.18.1049}, but valid in greater generality. The
recent work \cite{KuOtSp13} treats the one-dimensional case, see
Remark~\ref{remark:otte:first} below.

It is an open problem to prove the asymptotics \eqref{catastrophe} for
any $d\in\N$. There are even no known lower bounds on $\boldsymbol\langle
\Phi^N_{L},\Psi^N_{L}\boldsymbol\rangle$ except for $d=1$ \cite{KuOtSp13}.
A related problem,
recently solved in \cite{FrLeLeSe11}, concerns an effective estimate of the minimum change in total energy of the (infinite-volume) Fermi gas when a local, one-body potential is added to the kinetic energy.

The plan of this paper is as follows. We formulate our results in the next section. Sections~\ref{sec:proof-main} to~\ref{sec:proofError2} contain the proof of the main results, Theorems~\ref{thm:main} and~\ref{thm:main}'.
In the Appendix we prove Theorem~\ref{prop:scattering} and Corollary~\ref{corollary:anderson}, which relate  the diagonal of the Lebesgue density of a spectral correlation measure with the scattering matrix and the cross section. This yields a scattering-theoretic interpretation of our decay exponent $\gamma$ and shows that it coincides with that of \cite{PhysRevLett.18.1049}.

%
\section{Main results}
\label{sec:main}

We consider a pair of one-particle
Schr\"odinger operators $H_L:=-\Delta_L+ V_0$ and $H_L' := H_L + V$, which are self-adjoint and
densely defined in the Hilbert space $L^2(\Lambda_L)$, where $\Lambda_L
:= L\cdot\Lambda_1$ is obtained from scaling some fixed bounded domain $\Lambda_1\subseteq\R^d$,
which contains the origin, by a factor $L > 0$. The negative Laplacian
 $-\Delta_L$ is supplied with Dirichlet boundary conditions on $\Lambda_L$.
The background potential $V_0$ and the perturbation $V$ act as multiplication operators on
$L^2(\Lambda_L)$. They correspond to real-valued functions on $\R^{d}$ (denoted by the same letter) with the properties
\begin{equation}
	\label{def:potential}
	\tag{\textbf{A}}
  \begin{aligned}
  & \max\{V_{0},0\}\in K^d_\loc(\R^d),\
  \max\{-V_{0},0\}\in K^{d}(\R^d),\\
  &V \in L^{\infty}(\R^{d}), V \ge 0,\ \supp(V) \subseteq\Lambda_1
  \ \textnormal{compact}.
  \end{aligned}
\end{equation}
Here, we have written  $ K^{d}(\R^d)$ and $K^d_\loc(\R^d)$
for the Kato class and the local Kato class, respectively \cite{MR670130}.

We denote the self-adjoint, infinite-volume operators on $L^{2}(\R^{d})$,
corresponding to $H_{L}$ and $H_{L}'$, by $H$ and $H'$. Birman's theorem
\cite[Thm.\ XI.10]{MR529429}, together with \cite[Thm.\ B.9.1]{MR670130},
guarantees the existence and completeness of the
wave operators for the pair $H$, $H'$. In particular, the perturbation
$V$ does not change the absolutely continuous spectrum, i.e.
\begin{equation}
	\label{acac}
 	\sigma_{\ac}(H) = \sigma_{\ac}(H').
\end{equation}

Assumption \eqref{def:potential} ensures that the one-particle operators $H_L$ and $H_L'$ in finite volume are bounded from below and have purely discrete spectrum (which follows, e.g., from the fact that the semigroup operators $\exp\{-t H_{L}^{(\prime)}\}$ are trace class \cite[Thm.~6.1]{MR1756112} for each $t>0$). We write $\lambda_1^L\le\lambda_2^L\le\dotsb$ and
$\mu_1^L\le\mu_2^L\le\dotsb$
for their non-decreasing sequences of eigenvalues, counting multiplicities, and $(\varphi_j^L)_{j\in\N}$ and
$(\psi_k^L)_{k\in\N}$ for the corresponding
sequences of normalised eigenfunctions with an arbitrary choice of basis vectors in any
eigenspace of dimension greater than one.

The induced non-interacting $N$-particle Schr\"odinger operators
$\mathbold{H}_L$ and $\mathbold{H}_L'$ in finite volume act on the totally antisymmetric subspace $\bigwedge_{j=1}^N L^2(\Lambda_L)$ of the $N$-fold tensor product space and are given by
\begin{equation}
  \mathbold{H}^{(\prime)}_L
  :=
  \sum_{j=1}^N \one\otimes\dotsm\otimes \one \otimes H_L^{(\prime)} \otimes \one\otimes\dotsm \otimes \one,
\end{equation}
where the index $j$ determines the position of $H_L^{(\prime)}$ in the
$N$-fold tensor product of operators and $N\in\N$. The corresponding
ground states are the totally antisymmetrised products
\begin{equation}
   \Phi^N_L
  :=
  \frac{1}{\sqrt{N!}}\varphi_1^L\wedge\dotsm\wedge\varphi_N^L,
  \qquad
  \Psi^N_L :=
  \frac{1}{\sqrt{N!}}\psi_1^L\wedge\dotsm\wedge\psi_N^L.
\end{equation}
In order to avoid ambiguities from possibly degenerate eigenspaces and to realise a
given Fermi energy $E\in\R$ in the thermodynamic limit, we choose the number of particles as
\begin{equation} \label{eq:N}
  N = N_L(E) := \# \{\, j\in\N\st \lambda_j^L \le E \,\} \in\N_{0}.
\end{equation}
This choice turns out to be particularly simple from a technical point
of view.\footnote{See Lemma~\ref{lemma:ssf}(i). In the case of
Lemma~\ref{lemma:ssf}(ii), other choices of $N$, which lead to the same Fermi
energy, can easily be handled, too.}

We will be interested in the ground-state overlap
\begin{equation}
	\SO_L(E) := \boldsymbol{\Big\<}\Phi_L^{N_L(E)}, \Psi_L^{N_L(E)}\boldsymbol{\Big\>}_{N_L{(E)}}
   = \det\Bigparens{\<\varphi_j^L, \psi_k^L\>}_{j,k=1,\dotsc,N_L(E)},
  \label{def:overlap}
\end{equation}
asymptotically as $L\to\infty$.  Here, $\boldsymbol\<\dotid,\dotid\boldsymbol\>_{N}$ stands for the scalar product on the $N$-fermion space $\bigwedge_{j=1}^N L^2(\Lambda_L)$, where $N \in\N$,
and
 $\<\dotid,\dotid\>$ for the one on the single-particle space $L^{2}(\Lambda_{L})$. If $N_{L}(E)=0$, we set $\SO_{L}(E) := 1$.

\begin{remark}
	By our choice \eqref{eq:N} the particle density $\varrho$ of the two non-interacting fermion systems in the thermodynamic limit equals the integrated density of states
	\begin{equation}
		\label{rho-density}
 		\varrho = \lim_{L\to\infty} \frac{N_{L}(E)}{L^{d} \abs{\Lambda_{1}}}
	\end{equation}
	of the single-particle Schr\"odinger operator $H$ (or equivalently that of $H'$), provided the limit exists.	Here, $\abs{\dotid}$ denotes the Lebesgue measure on $\R^{d}$. For example, the limit \eqref{rho-density} exists if $V_{0}$ is periodic or vanishes at infinity. If the limit
	\eqref{rho-density} does not exist, then there must be more than one accumulation point because
	$\limsup_{L\to\infty}N_{L}(E)/L^{d} < \infty$ for every $E\in\R$ due to assumptions \eqref{def:potential}. But even in this case it makes still sense to study the asymptotic behaviour of the overlap $\SO_{L}(E)$ as $L\to\infty$.
\end{remark}

\noindent
The main result of this paper is an upper bound on the ground-state overlap $\SO_{L}(E)$
for large $L$. Throughout we use the convention $\ln 0 := -\infty$.

\begin{theorem} \label{thm:main}
  Assume conditions \eqref{def:potential} and let $(L_n)_{n\in\N} \subset\R_{\ge 0}$
  be a sequence of increasing lengths with $L_n\upto\infty$.
	Then there exists
	a subsequence $(L_{n_k})_{k\in\N}$ and a Lebesgue null
  set $\Nn \subset \R$ of exceptional Fermi energies	such that for every
  $E\in\R\setminus\Nn$ the ground-state overlap \eqref{def:overlap} obeys
  \begin{equation}
    \limsup_{k\to\infty}\frac{\ln\abs{\SO_{L_{n_k}}(E)}}{\ln L_{n_{k}}} \le - \frac{\gamma(E)}{2}
  \end{equation}
  with some decay exponent $\gamma(E) \ge 0$.
\end{theorem}

\begin{remarks}
\item
	The proof of Theorem~\ref{thm:main} follows from Lemma~\ref{lemmaoverlap} and Theorem~\ref{thm:AI} in the next section.
	In fact, we prove the slightly stronger statement
	\begin{equation} \label{eq:strongerstatement}
    \abs{\SO_{L_{n_k}}(E)}
   	 \le \exp\big[-\tfrac{a}{2}\,\gamma(E)\ln L_{n_k} + \oh(\ln L_{n_k})\big]
	  = L_{n_k}^ {-a\gamma(E)/2+\oh(1)}
  \end{equation}
  as $k\to\infty$ for every $0<a<1$ with an $a$-dependent error term.
\item
	Of course, Theorem~\ref{thm:main} is only interesting if the decay
	exponent $\gamma(E)$ is strictly positive. It emerges as the diagonal value of the
	Lebesgue density
	\begin{equation}
 		\gamma(E) =
                \lim_{\varepsilon\to 0} \frac{1}{\varepsilon^2}
                \mu_\ac\Bigparens{[E-\varepsilon/2,E+\varepsilon/2]^2}
                = \frac{\d\mu_{\ac}(E,E')}{\d (E,E')}\bigg|_{E'=E}
	\end{equation}
  of a spectral correlation measure, which is defined by
  \begin{equation}
    \mu_\ac(B\times B')  :=  \tr\Big(\sqrt{V}\1_{B}(H_\ac) V\1_{B'}(H_\ac') \sqrt{V}\Big),
		\qquad B, B' \in\Borel(\R),
	\end{equation}
        where $\1_{B}$ stands for the indicator function of a set $B$. We refer to
	Definition~\ref{def:central} and Lemma~\ref{lemma:lpdiagonal}
	below for further notations and details.
 	In particular, we have $\gamma(E) =0$, whenever $E\notin\sigma_{\ac}(H)$.
	We refer to
	Theorem~\ref{prop:scattering} and Corollary~\ref{corollary:anderson}
	for a scattering-theoretic interpretation of $\gamma(E)$.
\item We expect the result of Theorem~\ref{thm:main} to hold true also
      for sign-indefinite non-compactly supported perturbations $V$
	which decay sufficiently fast at infinity.
\item
	Anderson \cite{PhysRevLett.18.1049} treats the special case $d = 3$, $V_{0}=0$ and $V$
	spherically symmetric and argues that $\SO_{L}(E) = \Oh(L^{-\gamma(E)/2})$ as $L\to\infty$
	for $E>0$ with the same decay exponent $\gamma(E)$ as  in this paper.
	Thus, our theorem reproduces and generalises Anderson's informal computation.
	We note that there is a factor of 2 missing in the final result (7) in \cite{PhysRevLett.18.1049},
	which was apparently forgotten.
\item \label{remark:otte:first} The only other mathematical work dealing with AOC is the
      preprint \cite{KuOtSp13}. There, the special case $d=1$ and
      $V_0=0$ is treated. Moreover,
	the perturbation $V$ needs to be small in a certain sense.
	But it is not required to be non-negative, nor to be of
	compact support -- sufficiently fast decay is enough. In this
	context, \cite{KuOtSp13} prove a bound like
	\eqref{eq:strongerstatement} with $a=1$, and with the same decay
	exponent~$\gamma$ as in this paper. They also provide a lower bound on
	$\SO_L(E)$ with a smaller decay exponent \cite[Cor.~5.6]{KuOtSp13}.
\item \label{no-subseq}
  The reason for passing to a subsequence $(L_{n_{k}})_{k\in\N}$ in Theorem~\ref{thm:main}
  originates from Lemma~\ref{lemma:ssf} below. What stands behind it is the lack of
  known a.e.-bounds on the
  finite-volume spectral shift function for the pair of operators $H_{L}, H_{L}'$, which hold
  uniformly in the limit $L\to\infty$. This unfortunate fact has been noticed many times
  in the literature, see e.g.\ \cite{MR2596053}, and the pathological behaviour of the spectral shift function found in
  \cite{MR908658} illustrates that this is a delicate issue.
  However, in certain special situations such a.e.-bounds are known, and our result can be
  strengthened. More precisely, we have
\end{remarks}

\noindent
  \textbf{Theorem~\ref{thm:main}'.}~ \emph{Assume the situation of Theorem~\ref{thm:main} with $d=1$,
  or replace the perturbation potential $V$ in Theorem~\ref{thm:main} by a finite-rank operator
  $V= \sum_{\nu=1}^{n} \langle\phi_{\nu}, \dotid\rangle \, \phi_{\nu}$
	with compactly supported $\phi_{\nu} \in L^{2}(\R^{d})$ for $\nu=1,\ldots,n$,
  or consider the lattice problem on $\Z^d$ corresponding to the situation in Theorem~\ref{thm:main}.
 Then the ground-state overlap \eqref{def:overlap} obeys
  \begin{equation}
    \limsup_{L\to\infty} \frac{\ln\abs{\SO_{L}(E)}}{\ln L} \le - \frac{\gamma(E)}{2}
  \end{equation}
  with some decay exponent $\gamma(E) \ge 0$ for Lebesgue-a.e.\ $E\in\R$.
  }
\medskip

\noindent
Next we turn to the already mentioned interpretation of the decay exponent in terms of quantities from scattering theory. Such a relation between the density of a spectral correlation measure and the scattering matrix or cross section may be of independent interest. In our case, this relation reveals non-trivial scattering as a mechanism leading to AOC.

\begin{theorem} \label{prop:scattering}
  Assume \eqref{def:potential} with $V_0 = 0$. Then the decay exponent $\gamma(E)$
  in Theorem~\ref{thm:main} reflects the amount of scattering caused by the perturbation and is given by
  \begin{equation}
  	\label{gamma-s-mat}
    \gamma(E) = \frac{E^{(d-1)/2}}{(2\pi)^{d+1}}
    \int_{\mathbb{S}^{d-1}}\!\d\varOmega(\omega)\;
    \sigma(E,\omega)
    = (2\pi)^{-2} \,\norm{S(E) - \one}_\HS^2
	\end{equation}
	for Lebesgue-a.e.\ $E \ge0$ and $\gamma(E) = 0$ for $E < 0$.
  Here, $\sigma(E,\omega)$ stands for the total scattering cross-section
  for the pair of operators $H, H'$
  on the energy shell corresponding to~$E$ 	with incident direction $\omega\in\mathbb{S}^{d-1}$
  and $\d\varOmega$ is the Lebesgue measure on the unit sphere $\mathbb{S}^{d-1}$ in $\R^{d}$.
  On the right-hand side $S(E): L^{2}(\mathbb{S}^{d-1}) \rightarrow
  L^{2}(\mathbb{S}^{d-1})$ denotes the scattering matrix  and
  $\norm{\dotid}_\HS$ the Hilbert-Schmidt norm.
\end{theorem}

\begin{remarks}
\item	For convenience of the reader, we give a proof of the theorem
     	in the Appendix using generalised eigenfunctions.
\item 	We refer to
	\cite{MR1774673} for precise definitions of the scattering-theoretic quantities.
	We suspect that the theorem remains true for general Kato decomposable background
	potentials $V_{0}$, and also under the conditions of Theorem~\ref{thm:main}'. In fact, the relations in
	\cite[\S7]{BiEn67e}, which do not rely on generalised eigenfunctions, seem to indicate this.
\end{remarks}

\noindent
In order to see that our findings agree with those of Anderson \cite{PhysRevLett.18.1049},
we further specialise to $d=3$ dimensions and a spherically symmetric perturbation $V$.
\begin{corollary} \label{corollary:anderson}
  Let $d=3$. Assume \eqref{def:potential} with $V_0 = 0$ and $V$ spherically symmetric.
	Then the decay exponent $\gamma(E)$ in Theorem~\ref{thm:main} is given by
  \begin{equation}
    \gamma(E) = \frac{1}{\pi^{2}}\, \sum_{\ell=0}^\infty (2\ell + 1)\big(\sin\delta_\ell(E)\big)^2
  \end{equation}
  for Lebesgue-a.e.\ $E \ge0$ and $\gamma(E) = 0$ for $E < 0$.
  Here, $\delta_\ell(E)$, $\ell\in\N_{0}$, are the scattering phases.
\end{corollary}

\begin{remarks}
\item The proof of Corollary~\ref{corollary:anderson} will also be given in the
Appendix.
\item We refer to  \cite[XI.8.C]{MR529429} for a definition of the scattering phases.
\end{remarks}

%
\section{Proof of Theorems \ref{thm:main} and \ref{thm:main}'} \label{sec:proof-main}

We start by estimating the ground-state overlap in the same way as in the first step of
\cite{PhysRevLett.18.1049}.

\begin{lemma}\label{lemmaoverlap}
  For every length $L > 0$ and every Fermi energy $E\in\R$ we define the \emph{Anderson integral}
	\begin{align}
		\label{eq:AndInt}
  	\I_L(E) := {} & \sum_{j=1}^{N_L(E)} \sum_{k=N_L(E)+1}^\infty \bigabs{\<\varphi_j^L,\psi_k^L\>}^2
	\end{align}
  and obtain the estimate
	\begin{equation}
		\abs{\SO_L(E)} \le \exp[-\tfrac{1}{2}\I_L(E)].
	\end{equation}
\end{lemma}

\begin{proof}
  The assertion is true by definition if $N_L(E) = 0$. For $N_L(E)\ge 1$,
  we use Hadamard's inequality
  \begin{equation}
    \abs{\det M}\le \prod_{j=1}^n \abs{m_j}_2
  \end{equation}
  for an $n\times n$-matrix $M$ given by its column
  vectors $m_1,\dotsc,m_n\in\C^n$, where $\abs{\dotid}_2$ denotes the
  Euclidean norm. This gives
  \begin{equation}
    |\SO_L(E)| = \Big| \det\bigparens{\<\varphi_j^L, \psi_k^L\>}_{j,k=1,\dotsc,N_L(E)} \Big|
    \le \prod_{j=1}^{N_{L}(E)} \biggparens{\sum_{k=1}^{N_{L}(E)}
    		\bigabs{\<\varphi_j^{L},\psi_k^{L}\>}^2}^{\frac{1}{2}}
 	\end{equation}
  and therefore
	\begin{equation}
    \ln\abs{\SO_L(E)} \le \frac{1}{2} \sum_{j=1}^{N_{L}(E)} \ln
      \biggparens{\sum_{k=1}^{N_{L}(E)} \bigabs{\<\varphi_j^{L},\psi_k^{L}\>}^2}.
 \end{equation}
 Parseval's identity $\sum_{k=1}^{N_{L}(E)} \bigabs{\<\varphi_j^{L},\psi_k^{L}\>}^2
   = 1 - \sum_{k=N_{L}(E)+1}^\infty \bigabs{\<\varphi_j^{L},\psi_k^{L}\>}^2$ and the elementary
 inequality $\ln(1 + x) \le x$ for $x \ge -1$ then yield the claim of
 the lemma.
\end{proof}

\noindent
In order to define the decay exponent $\gamma(E)$ of the main theorem,
we need a convergence result due to Birman and \`Entina.
We write $H_\ac^{(\prime)}$ to denote the restriction of the operator
$H^{(\prime)}$ to its absolutely continuous subspace.

\begin{proposition}[\protect{\cite[Lemma 4.3]{BiEn67e}}]
	\label{prop:Birman}
	Assume the situation of Theorem~\ref{thm:main} or Theorem~\ref{thm:main}'.
  For $E\in \R$ and $\varepsilon >0$ we define the spectral
  projections
    \begin{align}
      P_E^\varepsilon :=
      \1_{]E-\varepsilon/2,E+\varepsilon/2[}(H_\ac),
      \qquad
      \varPi_{E}^\varepsilon :=
      \1_{]E-\varepsilon/2,E+\varepsilon/2[}(H'_\ac).
  \end{align}
  Then there exists a Lebesgue null set  $\Nn_0\subset\R$ such that the limits
  \begin{align} \label{eq:BirmanPPi}
    P_E :=   \lim_{\varepsilon\to 0}
    \frac{1}{\varepsilon}\sqrt{V}P_E^\varepsilon\sqrt{V},
    \qquad
    \varPi_{E} :=
    \lim_{\varepsilon\to 0}
    \frac{1}{\varepsilon}\sqrt{V} \varPi_{E}^\varepsilon\sqrt{V}
  \end{align}
  exist in trace class for all $E\in\R\setminus\Nn_0$ and define non-negative trace class operators
      $P_E$ and $\varPi_E$.
\end{proposition}

The above proposition guarantees that the quantities introduced in the
first part of the next definition are well-defined.

\begin{definition}
	\label{def:central}
  \begin{nummer}
  \item For $E,E' \in \R\setminus\Nn_0$ we introduce
 		\begin{equation}
 				\gamma_{1}(E) := \tr P_{E},\qquad
				\gamma_{2}(E) := \tr \varPi_{E}
                \end{equation}
        as well as the two-dimensional quantity
		\begin{equation}
				\gamma^{(2)}(E,E') := \tr (P_{E} \varPi_{E'})
                \end{equation}
        and its value on the diagonal
                \begin{equation}
				\gamma(E) := \gamma^{(2)}(E,E) \label{gamma-def}.
		\end{equation}
	  This gives rise to functions $\gamma_{1},\gamma_{2},\gamma : \R \rightarrow \R$ and
	  $\gamma^{(2)}: \R^{2} \rightarrow\R$ by setting them to zero
	  if the limits in \eqref{eq:BirmanPPi} do not exist.
  \item
  	The Borel measures $\mu_{\ac}^1$ and $\mu_{\ac}^2$ on $\R$ are
    defined by
    \begin{equation} \label{eq:1dmeasures}
      \mu_{\ac}^1(B) := \tr(\sqrt{V}\1_B(H_\ac)\sqrt{V}),
      \qquad
      \mu_{\ac}^2(B) := \tr(\sqrt{V}\1_B(H_\ac')\sqrt{V})
    \end{equation}
    for $B\in\Borel(\R)$.
  \item   
		The \emph{spectral correlation measure} $\mu_\ac$ on $\R^2$
    is defined by
    \begin{equation}
      \label{ac-measure}
      \mu_\ac(B\times B') :=
      \tr\Big(\sqrt{V}\1_B(H_\ac)V\1_{B'}(H_\ac') \sqrt{V}\Big)
    \end{equation}
    for $B,B'\in\Borel(\R)$.
  \end{nummer}
\end{definition}

\begin{remark}
	\label{1D-density}
  Both expressions in \eqref{eq:1dmeasures} define Borel
  measures by
  \cite[Thm.~B.9.2]{MR670130}. They are absolutely continuous
  with respect to Lebesgue measure on $\R$ and the functions $\gamma_{1}$, resp.\ $\gamma_{2}$
  are representatives of the Lebesgue densities of $\mu_{\mathrm{ac}}^{(1)}$,
  resp.\ $\mu_{\mathrm{ac}}^{(2)}$. In particular, $\gamma_{1}, \gamma_{2} \in L^{1}_{\loc}(\R)$.
\end{remark}

A corresponding statement for the two-dimensional measure $\mu_{\mathrm{ac}}$ is contained in

\begin{lemma}
	\label{lemma:lpdiagonal}
  The Borel measure $\mu_{\ac}$ is well-defined and
  absolutely continuous with respect to Lebesgue measure on $\R^{2}$.
  The function $\gamma^{(2)} \in L^{1}_{\loc}(\R^{2})$ is a representative of its
  	Lebesgue density and obeys
    \begin{equation} \label{eq:gamma}
      \gamma^{(2)}(E,E') \le \gamma_1(E) \gamma_2(E')
      \qquad \text{for all $E,E'\in\R$.}
    \end{equation}
\end{lemma}

\begin{proof}[Proof of Lemma \ref{lemma:lpdiagonal}]
  Let $B,B' \in \Borel(\R)$ be bounded. Then the non-negative
      expression \eqref{ac-measure} is finite
  because of \cite[Thm.~B.9.1]{MR670130}.
  Thus, it gives rise to a uniquely defined Borel measure
  on $\R^{2}$ by standard reasoning \cite[Thm.\ 23.3]{Bau01}.

  H\"older's inequality and the norm inequality
  $\|\dotid\|_{\ell^{2}} \le \|\dotid\|_{\ell^{1}}$ for the (standard) sequence spaces
  imply the estimate
  \begin{multline}
  	\label{hoelder}
    \tr \Big(\sqrt{V}\1_B(H_\ac) V \1_{B'}(H_\ac') \sqrt{V} \Big)
    \\ \le  \tr \big(\sqrt V \1_B(H_\ac) \sqrt V\big)\,
         \tr\big(\sqrt V \1_{B'}(H_\ac')  \sqrt V\big).
  \end{multline}
	The inequality \eqref{eq:gamma} follows directly from it. In turn, \eqref{eq:gamma} and Remark~\ref{1D-density}
	imply
	$\gamma^{(2)} \in L^{1}_{\loc}(\R^{2})$.

	To show absolute continuity of $\mu_{\ac}$, we conclude from \eqref{hoelder}
	and Remark~\ref{1D-density} that
  \begin{equation}
  	\label{product-bound}
 		\mu_{\mathrm{ac}}(C) \le \int_{C}\!\d E\d E'\, \gamma_{1}(E) \gamma_{2}(E')
	\end{equation}
  holds for all product sets $C= B\times B'$ with $B,B' \in \Borel(\R)$.
  The comparison theorem \cite[Thm.~II.5.8]{Els05} extends
  \eqref{product-bound} to all $C \in \Borel(\R^{2})$. In
  particular, $\mu_{\ac}$ is absolutely continuous
  with respect to two-dimensional Lebesgue measure.

	Due to absolute continuity of $\mu_\ac$ the limit
  \begin{equation}
    \lim_{\varepsilon\to 0}
    \varepsilon^{-2} \mu_{\ac}
    \Bigparens{[E-\varepsilon/2,E+\varepsilon/2] \times [E'-\varepsilon/2,E'+\varepsilon/2]}
    =
    \frac{\d\mu_\ac(E,E')}{\d(E,E')}
  \end{equation}
  exists for Lebesgue-a.e.\ $(E,E')\in\R^2$.
 	But, by definition, the left-hand side equals $\gamma^{(2)}(E,E')$
 for all $E,E'\in\R\setminus\Nn_0$.
  \end{proof}

\begin{remark}
  We work with a particular representative of the Lebesgue density
  of $\mu_\ac$ because we are interested in diagonal values
  $\gamma^{(2)}(E,E)=\gamma(E)$ of the density.
\end{remark}

\noindent
Theorems~\ref{thm:main} and \ref{thm:main}' will follow from Lemma~\ref{lemmaoverlap} and

\begin{theorem} \label{thm:AI}
	\begin{nummer}
	\item
  	Assume conditions \eqref{def:potential} and let $(L_n)_{n\in\N} \subset\R_{\ge 0}$ be a
  	sequence of increasing lengths with $L_n\upto\infty$. 	Then there exists
		a subsequence $(L_{n_k})_{k\in\N}$ and a Lebesgue null
  	set $\Nn \subset \R$ of exceptional Fermi energies	such that for every
  	$E\in\R\setminus\Nn$ the Anderson integral \eqref{eq:AndInt} obeys
		\begin{equation}
  		\I_{L_{n_k}}(E) \ge a\gamma(E) \,\ln L_{n_k} + \oh(\ln L_{n_k}) \qquad\quad \text{as~~} k\to\infty
		\end{equation}
		for every $ 0< a < 1$ and with $\gamma(E)$ given by
		\eqref{gamma-def}. Here, the error term depends on $a$.
	\item
		Assume the situation of Theorem~\ref{thm:main}'. Then there
  	exists a Lebesgue null set $\Nn \subset \R$ of exceptional Fermi energies,
  	such that for every $E\in\R\setminus\Nn$
  	the Anderson integral \eqref{eq:AndInt} obeys
		\begin{equation}
			\label{ILnice}
  		\I_{L}(E) \ge a\gamma(E) \,\ln L + \oh(\ln L) \qquad\quad \text{as~~} L\to\infty
		\end{equation}
		for every $ 0< a < 1$ and with $\gamma(E)$ given by
		\eqref{gamma-def}. Here, the error term depends on $a$.
	\end{nummer}
\end{theorem}

\begin{remarks}
\item Theorem~\ref{thm:AI} follows immediately from Lemma~\ref{lemma:ssf} and Theorem~\ref{thm:modAI}.
\item \label{Otte}
  We are now in a position to expand Remark~\ref{remark:otte:first} on
  \cite{KuOtSp13}. They prove the exact asymptotics
  \[
    \I_L \sim \gamma\ln L \qquad\quad\text{as~~} L\to\infty
  \]
  with the same decay exponent $\gamma$, which extends our
  Theorem~\ref{thm:AI} in their particular case. Technically, it relies on the exact knowledge of the eigenvalues and eigenfunctions of the one-dimensional Laplacian in an interval and sophisticated explicit computations.
\end{remarks}

\noindent
The next lemma estimates the error arising from a modification of the Anderson integral
so that all energy levels up to, respectively from, the Fermi energy $E$ are taken into account.
This is where the spectral shift function enters.
It is only part~\itemref{ssf} of this lemma which forces us to pick a subsequence
$(L_{n_{k}})_{k\in \N}$ of the original sequence of lengths  $(L_{n})_{n\in\N}$.

\begin{lemma} \label{lemma:ssf}
 	\begin{nummer}
	\item \label{ssf}
		Assume \eqref{def:potential} and let $(L_n)_{n\in\N} \subset\R_{\ge 0}$ be a
		sequence of increasing lengths with $L_n\upto\infty$.
		Then
		there exists a subsequence $(L_{n_k})_{k\in\N}$ 
		such that for Lebesgue-a.e.\ Fermi energy $E\in\R$
  	\begin{equation}
  		\label{AI-fixedE}
      \Bigabs{\F_{L_{n_k}}(E) - \I_{L_{n_{k}}}(E)} = \oh(\ln L_{n_{k}})
  	\end{equation}
  	as $k\to\infty$. Here,
  	\begin{equation}
 			\F_{L}(E) := \tr\Big(1_{{]-\infty,E]}}(H_L)
                                            \1_{]E,\infty[}(H_L')\Big)
		\end{equation}
  	is the \emph{fixed-energy Anderson integral}.
	\item \label{ssf-prime}
		Assume the situation of Theorem~\ref{thm:main}'. Then
		\begin{equation}
 		  \sup_{L>1}\sup_{E\in\R} \Bigabs{\F_{L}(E) - \I_{L}(E)} 	< \infty.
		\end{equation}
	\end{nummer}
\end{lemma}

\begin{proof}
 	Given $L>0$ and $E\in\R$, we recall from \eqref{eq:N} that, by definition,
	\begin{equation}
 		\lambda_{N_{L}(E)}^{L} \le E \quad\text{and}\qquad
 		\lambda_{N_{L}(E)+1}^{L} > E,
	\end{equation}
        where we use the convention $\lambda_0^L := -\infty$.
	This allows to rewrite the Anderson integral as
	\begin{align}
  	\I_L(E)
  	&=
  		 \sum_{j=1}^{N_{L}(E)} \sum_{k=N_{L}(E)+1}^\infty \abs{\<\varphi_j^L,\psi_k^L\>}^2 \notag\\
  	&=
  		\tr\bigg\{1_{]-\infty,E]}(H_L) \,
			\Big( \1_{]E,\infty[}(H_L')
       - \sum_{k\in \{1, \ldots, N_L{(E)}\}\,\st \mu_{k}^{L} >E} \<\psi_k^L,\dotid\>\psi_k^L \Big)\bigg\}.
	\end{align}
	The number of terms in the above $k$-sum
	\begin{align}
	  \#\Big\{k\in \{1, \ldots, N_L{(E)}\} \st \mu_k^L >E\Big\}
	  & = N_{L}(E) - \#\Big\{k\in \N \st \mu_k^L \le E\Big\} \notag\\
          & =: \xi_{L}(E)
 	\end{align}
	is precisely the value at $E$ of the (non-negative) spectral shift function
	for the pair of finite-volume operators $H_{L},H_{L}'$. Therefore we obtain
	\begin{equation}
		\label{xi-bound}
 		 0 \le  \F_{L}(E) -  \I_{L}(E)
		 \le \xi_{L}(E),
	\end{equation}
  and it remains to prove that this error is of order $\oh(\ln L)$ as $L\to\infty$.
  In the situation of \itemref{ssf-prime}, we have even
  $\sup_{L>1}\sup_{E\in\R} \xi_{L}(E) < \infty$ thanks to a
  finite-rank argument and the min-max principle.
  In order to apply this finite-rank argument in the one-dimensional
  continuum case, use Dirichlet-Neumann bracketing and the fact that
  introducing a Dirichlet or Neumann boundary point amounts to a rank-two-perturbation for the resolvents.

  In the multi-dimensional continuum situation of \itemref{ssf} no
  such uniform bounds are known -- not even bounds for a.e.\ energy. But we
  can  exploit the weak convergence \cite[Thm.~1.4]{MR2596053}
  \begin{equation}
      \lim_{L\to\infty}\int_I\;\d E\, \xi_L(E)
      =
      \int_I\;\d E\, \xi(E),
  \end{equation}
  for every bounded interval $I\subseteq\R$, where
  $\xi\in L^1_\loc(\R)$ is the spectral shift function for the pair of infinite-volume operators $H,H'$.
  Thus, given a sequence of diverging lengths $(L_n)_{n\in\N}$,
  the sequence of non-negative functions $(\xi_{L_{n}}/\ln L_{n})_{n\in\N}$ converges to
  zero in $L^{1}(I)$. Hence there exists a subsequence $(L_{n_{k}})_{k\in\N}$ such that
  $(\xi_{L_{n_{k}}}/\ln L_{n_{k}})_{k\in\N}$ converges to zero for Lebesgue-a.e.\ $E\in I$.
  The claim then follows from exhausting $\R$ by a sequence of bounded intervals $I$.
\end{proof}

\begin{theorem}
	\label{thm:modAI}
	Assume the situation of Theorem~\ref{thm:main} or Theorem~\ref{thm:main}'.
	Then there exists a Lebesgue null set
	$\mathcal{N} \subset\R$ of exceptional Fermi energies such that for every
	$E \in\R\setminus\mathcal{N}$ and every $a \in{]0,1[}$
	\begin{equation}
  	\F_{L}(E) \ge a\gamma(E) \ln L  + \oh(\ln L)
	\end{equation}
	as $L\to\infty$. Here, the error term depends on $a$.
\end{theorem}

\noindent
We will explicitly spell out the proof of Theorem~\ref{thm:modAI} for the situation of
Theorem~\ref{thm:main} only. The proof is fully analogous (and even simpler) in the remaining  situations of Theorem~\ref{thm:main}', where $V$ is a finite-rank operator or that of the lattice model.

In the first lemma which enters the proof of Theorem~\ref{thm:modAI} we rewrite the fixed-energy Anderson integral as an integral with respect to a spectral correlation measure.

\begin{lemma} \label{lemma:AIasIntegral}
	Assume \eqref{def:potential}, let $L>0$ and $E\in\R$. Then we have
  \begin{align}\label{integralausdruck}
    \F_{L}(E)
    &= \int_{]-\infty,E]\times]E,\infty[}\frac{\d\mu_L(x,y)}{(y-x)^2} \nonumber\\
    &  \ge \int_0^{L^a} \dt\; t \int_{\R^2} \d\mu_L(x,y)\; \e^{-t(y-x)}
    	\chi_L^-(x) \chi_L^+(y),
  \end{align}
  where the \emph{(finite-volume) spectral correlation measure} $\mu_{L}$ on $\R^{2}$ is
  uniquely defined by $\mu_L(B\times B') := \tr\big(\sqrt{V} \1_B(H_L) V \1_{B'}(H_L') \sqrt{V} \big)$
	for $B,B'\in\Borel(\R)$. The parameter $a>0$ and the functions $\chi_L^\pm \in L^{\infty}(\R)$
	are arbitrary subject to
	\begin{equation}
		\label{chi-basic}
 		0 \le \chi_L^+ \le \1_{]E,\infty[}\quad \text{and}\qquad 0 \le \chi_L^- \le \1_{]-\infty, E]}.
	\end{equation}
\end{lemma}

\begin{remark}
We have suppressed the dependence of $\chi_{L}^{\pm}$ on the Fermi energy $E$ and will impose further properties on these functions in Definition~\ref{sc-def} below.
\end{remark}

\begin{proof}[Proof of Lemma~\ref{lemma:AIasIntegral}]
   	The eigenvalue equations imply
    \begin{equation}
      \lambda_j^L\<\varphi_j^L, \psi_k^L\>
      =
      \<H_L \varphi_j^L, \psi_k^L\>
      =
      \mu_k^L\<\varphi_j^L, \psi_k^L\> - \<\varphi_j^L, V\psi_k^L\>
    \end{equation}
    from which we obtain the identity
    \begin{equation}
      \abs{\<\varphi_j^L, \psi_k^L\>}^2
      =
      \frac{\abs{\<\varphi_j^L, V\psi_k^L\>}^2}{(\mu_k^L - \lambda_j^L)^2} \,,
    \end{equation}
  	provided  $\lambda_j^L \neq \mu_k^L$.
    This yields
    \begin{align}
      \F_{L}(E)
      = \sum_{\substack{j\in\N\,:\\ \lambda_j^L \le E}}
      \sum_{\substack{k\in\N\,:\\ \mu_k^L > E}}
      \abs{\<\varphi_j^L, \psi_k^L\>}^2
      & =
      \sum_{\substack{j\in\N\,:\\ \lambda_j^L \le E}}
      \sum_{\substack{k\in\N\,:\\ \mu_k^L > E}}
      \frac{\abs{\<\varphi_j^L, V\psi_k^L\>}^2}{(\mu_k^L - \lambda_j^L)^2}
      \nonumber \\
      & =
      \int_{]-\infty,E]\times]E,\infty]}\frac{\d\mu_L(x,y)}{(y-x)^2}.
    \end{align}
		The inequality in \eqref{integralausdruck} follows from the integral representation
		$x^{-2} = \int_{0}^{\infty}\d t\, t \, \e^{-tx}$ for $x >0$, Fubini's theorem, from cutting
		the $t$-integral and \eqref{chi-basic}.
\end{proof}

\begin{definition}
	\label{sc-def}
	Given an exponent $b >0$, a length $L>1$, a cut-off energy $E_{0} \ge 1$ and a Fermi energy
	$E\in [-E_{0}+1,E_{0} -1]$, we say that
	$\chi_{L}^{\pm} \in C_{c}^{\infty}(\R)$
	are \emph{smooth cut-off functions}, if they obey
	\begin{equation}
		\begin{split}
  		1_{[E+ L^{-b}, E_{0}]} &\le \chi_{L}^{+} \le 1_{]E , E_{0}+1[}, \\
			1_{[-E_{0},E - L^{-b}]}  &\le \chi_{L}^{-} \le 1_{]-E_{0}-1,E [}
		\end{split}
	\end{equation}
	and if there exist $L$-independent constants $c_{\nu}
	\in \R_{>0}$, $\nu\in\N_{0}$, such that
	\begin{equation} \label{def:chi_L2}
 		\chi_L^\pm (E \pm x)\le c_{0} L^{b} \,x
	\end{equation}
	for all $x\in[0, L^{-b}[$ and
	\begin{equation} \label{def:chi_L3}
		\biggabs{\frac{\partial^\nu}{\partial x^\nu}\, \chi^\pm_L(E \pm x)}   \le
		\begin{cases} \; c_\nu L^{\nu b} & \text{for all}\;\; x\in[0, L^{-b}[ \,, \\[.5ex]
 				\; c_{\nu} & \text{otherwise}
	  \end{cases}
	\end{equation}
	for every $\nu\in\N$.

	Thus, $\chi_{L}^{+}$ equals one inside $[E +L^{-b}, E_{0}]$ and zero in ${]}-\infty, E] \cup [E_{0} +1,\infty[$.
	Whereas its smooth growth in $[E, E+ L^{-b}]$ gets steeper with increasing $L$, we choose its
	smooth decay in $[E_{0},E_{0}+1]$ independently of $L$. The
	properties of $\chi_{L}^{-}$ are analogous.
\end{definition}

\noindent
The next lemma allows to replace the finite-volume operators by their infinite-volume analogues.
It is the crucial step in our argument, and we defer
the proof to Sect.~\ref{Section:Error:Proof}.

\begin{lemma}\label{Lemma:Error}
 	Let $0<a <b< 1$, $L>1$ and $E_{0} \ge 1$. Pick a Fermi energy $E\in [-E_{0}+1,E_{0} -1]$
	and let $\chi_{L}^{\pm}$ be the associated smooth cut-off functions.
	Then we have
 	\begin{align}
  	\int_0^{L^a} \dt\; t\, &  \bigg( \int_{\R^2} \d\mu_L(x,y) \,
			\e^{-t(y-x)} \chi_L^-(x) \chi_L^+(y) \nonumber\\
	   & \hspace{1em} - \int_{\R^2} \d\mu(x,y)\, 	\e^{-t(y-x)} \chi_L^-(x) \chi_L^+(y) \bigg)
   	= \oh(1)
 	\end{align}
 	as $L\to\infty$, where the $\oh(1)$-term depends on $a$ and
 	$b$. Here, the \emph{(infinite-volume) spectral
 	correlation measure} $\mu$ on $\R^{2}$ is uniquely defined by
	\begin{equation}
 		\mu(B\times B') := \tr\big(\sqrt{V} 1_B(H) V 1_{B'}(H') \sqrt{V} \big)
	\end{equation}
	for $B,B'\in\Borel(\R)$.
\end{lemma}

\noindent
We recall the measure $\mu_{\ac}$ from Definition~\ref{def:central} and
the Lebesgue densities $\gamma_{1}$ and $\gamma_{2}$ of the measures $\mu_{\mathrm{ac}}^{(1)}$ and $\mu_{\mathrm{ac}}^{(2)}$, see
Remark~\ref{1D-density}. In the next lemma we estimate the error
for replacing the smooth cut-off functions in the limit expression
(more precisely, its ac-part) of the previous lemma by step functions.

\begin{lemma}\label{Lemma:Error2}
 	In addition to the hypotheses of the previous lemma, suppose that
	$E\in [-E_{0}+1,E_{0} -1]$ is a Lebesgue point of both $\gamma_{1}$ and $\gamma_{2}$.
	Then we have
 	\begin{equation}\label{lemma2error}
    \int_0^{L^a} \!\!\!\dt\, t \int_{\R^2}\!\!\d\mu_{\ac}(x,y)\, \e^{-t(y-x)}
     \Bigparens{\chi_{L}^-(x) \chi_{L}^+(y) - \1_{[-E_{0},E]}(x)\1_{[E,E_{0}]}(y)}
 		= \Oh(1)
	\end{equation}
 	as $L\to\infty$ with an $\Oh(1)$-term depending on $a$ and $b$.
\end{lemma}

\noindent
The proof of this lemma is deferred to Sect.~\ref{sec:proofError2}.

In the last lemma, we show how the diagonal of the $\mu_\ac$-density
arises in the large-$t$ limit.

\begin{lemma} \label{lemma:delta_convergence}
  For Lebesgue-a.e.\ $E\in [-E_0, E_0]$ we have
  \begin{align} \label{eq:delta_convergence}
    \lim_{t\to\infty}
    \int_{\R^2} \!\!\d\mu_\ac(x,y)\, t^2 e^{-t(y-x)}1_{[-E_0,E]}(x)1_{ [E,E_0]}(y)
    = \gamma(E).
  \end{align}
\end{lemma}
\begin{proof}
  To shorten formulas, we suppress the subscript $\ac$ and write
  $H^{(\prime)} = H^{(\prime)}_\ac$ in this proof.
  Recalling the definition of $\mu_\ac$ and the identity $\tr(P_E \varPi_E) =
  \gamma(E)$, which is valid for Lebesgue-a.e.\ $E\in\R$, we have to show that
  \begin{equation}
    \Bigabs{\tr\Bigbraces{\sqrt{V} t e^{t(H-E)}\1_{[-E_0,E]}(H)
                  V t e^{-t(H'-E)}\1_{[E,E_0]}(H') \sqrt{V}} -
    \tr(P_E\varPi_E)}
  \end{equation}
  vanishes as $t\to\infty$.
  We bound this expression from above by
 \begin{align} \label{eq:delta_convergence_inserted}
    &\Bigl\lvert\tr\Bigbraces{
      \Bigparens{\sqrt{V} t e^{t(H-E)}\1_{[-E_0,E]}(H)\sqrt{V} - P_E}
      \sqrt{V} t e^{-t(H'-E)}\1_{[E,E_0]}(H')\sqrt{V}}\nonumber \\
    &\quad\qquad\qquad+ \tr\Bigbraces{P_E \Big(\sqrt{V}te^{-t(H'-E)}\1_{[E,E_0]}(H')\sqrt{V}
      - \varPi_E \Big)} \Bigr\rvert\nonumber
    \\
    &\le \Bigparens{\sup_{t\ge 0}
      \bignorm{\sqrt{V} t e^{-t(H'-E)}\1_{[E,E_0]}(H')\sqrt{V}}}\nonumber \\
    &\qquad\qquad\qquad\qquad\qquad\qquad
    \times\tr\bigabs{\sqrt{V} t e^{t(H-E)}\1_{[-E_0,E]}(H)\sqrt{V} - P_E}\nonumber
    \\&\quad\qquad\qquad + \norm{P_E} \tr\bigabs{\sqrt{V}te^{-t(H'-E)}\1_{[E,E_0]}(H')\sqrt{V} - \varPi_E}.
  \end{align}
  First, we claim that
  \begin{equation}
  	\label{Pi-conv}
   \lim_{t\to\infty} \sqrt{V}te^{-t(H'-E)}\1_{[E,E_0]}(H')\sqrt{V} = \varPi_E
  \end{equation}
  in trace class, for Lebesgue-a.e.\ $E\in[-E_0,E_0]$. To
  see this, we show convergence of the trace norms and weak
  convergence, which implies convergence in trace class by
  \cite[Addendum H]{MR2154153}. For the trace norms, we compute
  \begin{align} \label{eq:delta_convolution}
    \tr\bigabs{\sqrt{V}te^{-t(H'-E)}\1_{[E,E_0]}(H')\sqrt{V}}
     &= \tr\bigparens{\sqrt{V}te^{-t(H'-E)}\1_{[E,E_0]}(H')\sqrt{V}}
     \nonumber\\
     &= \int_E^{E_0}\!\!\dy\, \gamma_2(y) t e^{-t(y-E)} \nonumber\\
     &= \bigparens{(\gamma_2\1_{[-E_0,E_0]}) * \varrho_t}(E),
  \end{align}
  with $x\mapsto\varrho_t(x) := t e^{tx} \1_{]-\infty,0[}(x)$ being an
  approximation of the Dirac delta distribution. As $t\to\infty$, the convolution in
  \eqref{eq:delta_convolution} converges for Lebesgue-a.e.\ $E\in[-E_0,E_0]$
  to $\gamma_2(E) = \tr \varPi_E = \tr\abs{\varPi_E}$,
    see e.g.\ \cite[Chap.~1]{MR1232192}. Thus, the trace norm of
  $\sqrt{V}te^{-t(H'-E)}\1_{[E,E_0]}(H')\sqrt{V}$ converges to that of
  $\varPi_E$. In particular,
  \begin{equation} \label{eq:sup_norm_texp}
    \sup_{t\ge 0} \bignorm{\sqrt{V} t
    e^{-t(H'-E)}\1_{[E,E_0]}(H')\sqrt{V}}
    < \infty.
  \end{equation}

  It remains to show weak convergence. To
  this end, take some dense countable set $\mathcal{D}\subseteq
  L^2(\R^d)$.
  Then by a similar delta-argument as above
  \begin{equation}
    \lim_{t\to\infty} \<\varphi, \sqrt{V}te^{-t(H'-E)}\1_{[E,E_0]}(H')\sqrt{V}\psi\>
    = \<\varphi, \varPi_E\psi\>
  \end{equation}
  for all $\varphi, \psi\in\mathcal{D}$ and all $E\in
  [-E_0,E_0]$ outside a null set depending on $\mathcal{D}$.
  Together with \eqref{eq:sup_norm_texp}, this proves weak
  convergence to $\varPi_E$ for Lebesgue-a.e.\ $E\in [-E_0,E_0]$, see
  \cite[Thm.~4.26]{MR566954}.

  The same argument proves $\lim_{t\to\infty} \sqrt{V}te^{t(H-E)}\1_{[-E_0,E]}(H)\sqrt{V} = P_{E}$
  in trace class. Using this, \eqref{eq:sup_norm_texp}, the
  boundedness of $P_E$ and \eqref{Pi-conv}, the right-hand side of
  \eqref{eq:delta_convergence_inserted} is seen to vanish as $t\to\infty$.
\end{proof}

We are now ready for the proof of Theorem~\ref{thm:modAI}, which also
completes the proof of Theorem~\ref{thm:AI}
and, thus, of Theorems~\ref{thm:main} and \ref{thm:main}'.

\begin{proof}[Proof of Theorem~\ref{thm:modAI}]
Let $0<a<b<1$. Lemmas~\ref{lemma:AIasIntegral}, \ref{Lemma:Error} and \ref{Lemma:Error2}
imply that
\begin{equation}
	\label{start-ineq}
	 \F_{L}(E)  \ge
	 \int_0^{L^a} \!\dt\, t  \int_{[-E_0,E]\times [E,E_0]} \!\d\mu_{\ac}(x,y) \,\e^{-t(y-x)} \,
	 \;+\; \Oh(1)
\end{equation}
as $L\to\infty$ for Lebesgue-a.e.\ $E\in [-E_{0}+1,E_{0} -1]$ (more precisely
those $E$ which are Lebesgue points of $\gamma_{1}$ and $\gamma_{2}$). Here, we have also used $\mu(C) \ge \mu_{\ac}(C)$ for every $C\in \Borel(\R^{2})$ and
the non-negativity of the integrand.

\noindent Furthermore, using Lemma~\ref{lemma:delta_convergence},
\begin{equation}
  \lim_{L\to\infty} \frac{1}{\ln L} \int_1^{L^a} \frac{\dt}{t} \;
  \biggabs{
    \int_{[-E_0,E]\times [E,E_0]} \!\!\d\mu_{\ac}(x,y) \,\ t^2 e^{-t(y-x)} \,
    - \gamma(E)} = 0
\end{equation}
for Lebesgue-a.e.\ $E\in [-E_0+1, E_0 - 1]$ (the exceptional set being
independent of $a$).
Thus, we can replace the integrand in \eqref{start-ineq} by
$\gamma(E)$ at the expense of a sublogarithmic error and arrive at
\begin{equation}
 	\F_{L}(E) \ge  \int_1^{L^a} \frac{\dt}{t} \; \gamma(E) + \oh(\ln L)
\end{equation}
as $L\to\infty$ for Lebesgue-a.e.\ $E\in\R$, which proves the theorem.
\end{proof}

%
\section{Proof of Lemma \ref{Lemma:Error}} \label{Section:Error:Proof}

To shorten the formulas, we assume w.l.o.g.\ that $E=0$. This can always be achieved by an energy
shift of the Hamiltonians. We define the abbreviations
\begin{equation}
  g_L^t(x)  := \chi_L^-(x)\,\e^{tx}
  \qquad \text{and}\qquad
  f_L^t(x)  := \chi_L^+(x)\,\e^{-tx}
\end{equation}
for every $x\in\R$ and $t \ge 0$ so that Lemma~\ref{Lemma:Error} can be reformulated as
\begin{equation}
	\label{kah}
  \int_0^{L^a} \!\dt\, t \,K_{L}(t)  = \oh(1)
\end{equation}
as $L\to\infty$ with
\begin{equation}
  K_{L}(t) := \tr\Bigparens{\sqrt V g_L^t(H_L)  V f_L^t(H_L')\sqrt V}
  	- \tr\Bigparens{\sqrt V g_L^t(H) V f_L^t(H')\sqrt V} .
  \label{errorbound2}
\end{equation}
We estimate this function according to $	|K_{L}(t)| \le K_{L}^{(1)}(t) + K_{L}^{(2)}(t)$, where
\begin{align}
 	 K_{L}^{(1)}(t) &:= \tr\left(\sqrt V f_L^t(H')\sqrt V\right)
	 	\Bignorm{\sqrt V \Big(g_L^t(H_L)-g_L^t(H)\Big)\sqrt V}, \\
  K_{L}^{(2)}(t) &:= \tr\left(\sqrt V g_L^t(H_{L})\sqrt V\right)
    \Bignorm{\sqrt V \Big(f_L^t(H'_L)-f_L^t(H')\Big)\sqrt V}.
	\label{errorbound}
\end{align}
Since both $K_{L}^{(1)}$ and $K_{L}^{(2)}$ can be estimated in the very same way,
we will demonstrate the argument for $K_{L}^{(2)}$ only.
Our main technical tool is the Helffer-Sj\"ostrand formula, see e.g.\ \cite[Section IX]{MR1768629}, according to which
\begin{equation}\label{Helffer-Sj}
  f_L^t(H_L')-f_L^t(H')
  = \frac 1 {2\pi}
    \int_\C \dz\;
    \big(\partial_{\bar z} \wtilde f_{L}^{t}(z)\big)\left(\frac 1 {H_L'-z}-\frac 1 {H'-z}\right).
\end{equation}
Here, $z:=x+iy$, $\partial_{\bar z}:=\partial_x+i\partial_y$, $\d z := \d x\d y$
and $\wtilde f_{L}^{t}\in C_c^2(\C)$ is an almost analytic extension of $f_L^t$ to the complex plane.
The latter can be chosen as
\begin{equation}
  \wtilde f_{L}^{t}(z):= \xi(z)
  \sum_{k=0}^{n} \frac {(iy)^k}{k!} \, \frac{\d^kf_L^t}{\d x^k}(x)
\end{equation}
for some $n\in\N$ and some $\xi\in C_c^{\infty}(\C)$ with $\xi(z) = 1$ for all $z\in
\supp f_L^t \times [-1,1]$, $ \xi(z) = 0$ for all $z$ such that
$\dist_{\C}(z,\supp f_L^t)\ge 3$ and $\xi(z) \in [0,1]$ otherwise.
We will assume $n \ge 2$ below.
Since $\supp  f_L^t = [0, E_{0}+1]$, the function $\xi$ can be chosen independently of $L$ and~$t$,
and such that ${\norm{\xi}}_{\infty}=1$ and ${\norm{ \xi'}}_{\infty}<1$.

For later purpose we introduce the function $h:=\sum_{k=0}^{n+1}\bigabs{\frac{\d^kf_L^t} {\d x^k}}
\in C_{c}(\R)$ and
infer the existence of a constant $C\in {]0,\infty[}$, which is independent of $L$ and $t$, such that
 \begin{equation}\label{def:chi_L5}
  \abs{\partial_{\bar z} \wtilde f_{L}^{t}(z)}\le C\abs{y}^n h(x)
 \end{equation}
for all $z\in\C$.
Furthermore, the bound \eqref{def:chi_L3} implies the estimate
\begin{equation}
 	\biggabs{\frac{\d^kf_L^t} {\d x^k}(x)} \le L^{bk} \sum_{\kappa=0}^{k} \bigg(
\begin{array}{@{}c@{}} k\\ \kappa\end{array} \bigg)
	\bigg(\frac{t}{L^{b}} \bigg)^{\kappa} c_{k-\kappa} \,\1_{[0, E_{0}+1]}(x)
\end{equation}
for every $t\ge0$, $L \ge 1$ and $x\in\R$. From this we conclude the
existence of a polynomial $Q_{n}$ over $\R$
of degree $n+1$ with non-negative coefficients such that
\begin{equation}
	0 \le h(x) \le Q_{n}(t/L^{b})\, L^{b(n+1)} \,\1_{[0, E_{0}+1]}(x).
  \label{def:chi_L7}
\end{equation}
 We will split the contribution of \eqref{Helffer-Sj} in \eqref{errorbound} into two parts.
 Accordingly, we define for $\varepsilon \in {]0,1-b[}$
\begin{equation}
	\label{D<}
  D_{L}^{<}(t) := \frac 1 {2\pi} \int_{\abs{y} \le L^{-1+\varepsilon}} \d z\,
    \big(\partial_{\bar z} \wtilde f_{L}^{t}(z) \big)\, \sqrt V\left[ \frac 1 {H'_L-z} - \frac 1 {H'-z}\right] \sqrt V
\end{equation}
and
\begin{equation}
	\label{D>}
  D_{L}^{>}(t):= \frac 1 {2\pi} \int_{\abs{y} > L^{-1+\varepsilon}} \d z\,
  	\big(\partial_{\bar z} \wtilde f_{L}^{t}(z)\big) \, \sqrt V\left[\frac 1 {H'_L-z}-\frac 1{H'-z}\right] \sqrt V .
\end{equation}
Then, using the boundedness of $V$ and the estimates \eqref{def:chi_L5} and \eqref{def:chi_L7},
we obtain
\begin{align}
	\label{Helf-Sjostrand2}
  \norm{D_{L}^{<}(t)}
  & \le
  \frac 1 {2\pi} \int_{\abs{y}\le L^{-1+\varepsilon}} \!\dz\;
  \abs{\partial_{\bar z} \tilde f(z)}\,
  \frac 2 {\abs{y}} \,\norm{\sqrt V}^2 \nonumber \\
  & \le \frac{C}{\pi} \, \|V\|_{\infty}
  	\int_{\abs{y}\le L^{-1+\varepsilon}} \!\d z \,\abs{y}^{n-1} {h(x)} \nonumber \\
  & = \frac{2C}{\pi n} \, \|V\|_{\infty}  L^{n(-1+\varepsilon)} \int_{\R} \!\d x\, h(x) \nonumber\\
  &	\le C_{<} \, Q_{n}(t/L^{b}) \, L^{b+n(-1+\varepsilon +b)} ,
\end{align}
where $C_{<}:= (2C/\pi n) \, (E_{0}+1)\|V\|_{\infty} < \infty$.

We estimate the norm of \eqref{D>} with the help of the geometric resolvent inequality --
see e.g.\ \cite[Lemma 2.5.2]{MR1935594}, whose proof extends to Kato decomposable potentials --, the bound \eqref{def:chi_L5} and the fact that $\xi(z) =0$ if
$\dist(z,\R) \ge 3$. This gives for $L > 3$
\begin{align}
	\label{De-ge}
  \norm{D_{L}^{>}(t)}
 & \le \frac{C_{\textsc{gre}}}{2\pi}  \, \|V\|_{\infty} \int_{\abs{y}>L^{-1+\varepsilon}} \dz\;
  \abs{\partial_{\bar z} \tilde f(z)}\Bignorm{\1_{\supp V} \,
  \frac 1{H'_L-z}\1_{\delta\Lambda_{L}}}  \notag \\
  & \hspace*{5cm} \times\Bignorm{\1_{\delta\Lambda_{L}}\frac 1 {H'-z} \1_{\supp V}}
  	\notag \\
  &
  \le    \frac{CC_{\textsc{gre}}}{2\pi}  \, \|V\|_{\infty}
  \int_{\abs{y} \in {]L^{-1+\varepsilon}, 3]}} \!\dz\, h(x) \,\abs{y}^{n-1}
  \Bignorm{\1_{\delta\Lambda_{L}} \frac 1 {H'-z}\1_{\supp V}},
\end{align}
where $\delta\Lambda_{L} := \Lambda_L\setminus \Lambda_{L-1}$ and the constant
$C_{\textsc{gre}} < \infty$ depends only on $E_{0}$, the space dimension and the
potentials $V_{0}$ and $V$.
The operator norm in the last line of \eqref{De-ge} is bounded by a
Combes-Thomas estimate for operator kernels of
resolvents, see e.g.\ \cite[Thm.~1]{MR1937430},
\begin{equation}
	\label{CoTho}
  \Bignorm{\1_{\varGamma} \frac 1 {H'-z} \1_{\varGamma'}}
  \le \frac{C_{\textsc{ct}}}{|y|} \, \e^{- c_{\textsc{ct}}\dist(\varGamma,\varGamma') \abs{y}}.
\end{equation}
It holds for all cubes $\varGamma, \varGamma' \subset\R^{d}$
of side length $1$ and all $z$ in some bounded subset of $\C$, which we choose as
$\supp(h) \times [-3,3]$. The constants $C_{\textsc{ct}}, c_{\textsc{ct}} \in {]0,\infty[}$ in \eqref{CoTho}
can be chosen to depend only on $E_{0}$, the space dimension and the
potentials $V_{0}$ and $V$. Now, we assume $n \ge 2$, cover $\supp(V)$ and $\delta\Lambda_{L}$ by unit cubes and apply the bounds \eqref{CoTho} and \eqref{def:chi_L7} to \eqref{De-ge}. In this way we infer the
existence of a constant $\wtilde{C}_{>} \in {]0,\infty[}$, which is independent of $L$ and $t$, such that
\begin{align}
	\norm{D_{L}^{>}(t)} &
  \le \wtilde{C}_{>} (E_{0}+1) Q_{n}(t/L^{b}) L^{d + b(n+1)}
  	\int^{3}_{L^{-1+\varepsilon}} \!\dy\; {y}^{n-2}
		\e^{-c_{\textsc{ct}} L y/2}  \nonumber \\
  & \le C_>   Q_{n}(t/L^{b}) \, L^{d + b(n+1)} \,
  \e^{-c_{\textsc{ct}} L^{\varepsilon}/2} \label{Helf-Sjostrand3}
\end{align}
for all $t \ge 0$ and all $L >L_{n}$, where $L_{n}$ depends only on $\supp(V)$ and
$C_{>} := 3^{n-1} (n-1)^{-1} \wtilde{C}_{>} (E_{0}+1) $.

Combining \eqref{errorbound}, \eqref{Helffer-Sj}, \eqref{D<} -- \eqref{Helf-Sjostrand2} and
\eqref{Helf-Sjostrand3}, we obtain the estimate
\begin{equation}
	\label{Helf-Sjo3b}
  \int_0^{L^a}\dt\, t  \, K_{L}^{(2)}(t)
  \le \Big( C_{<} L^{b+n(-1+\varepsilon+b)} + C_{>} L^{d+b(n+1)}
    	 \e^{-c_{\textsc{ct}} L^{\varepsilon}/2}\Big) \Phi_{n}(L) ,
\end{equation}
for all $L > L_{n}$, where
\begin{align}
	\label{Helf-Sjo4}
 	\Phi_{n}(L) &:= \int_0^{L^a}\!\dt\, t Q_{n}(t/L^{b}) \tr\left(\sqrt V g_L^t(H_L)\sqrt V\right)
	  \nonumber\\
	& \phantom{:}\le L^{a} Q_{n}(L^{a-b})  \int_0^{\infty}\dt\,
	\tr\left(\sqrt V g_L^t(H_L)\sqrt V\right)
\end{align}
The map $\nu_L: B\mapsto \tr\big(\sqrt V \1_B(H_L)\sqrt V\big)$, where $B\in\Borel(\R)$,
defines a Borel measure on $\R$. In this context we note that
\begin{equation}
	\label{fin-int}
 	\what{\nu}(B) := \sup_{L>1}\nu_{L}(B) < \infty \qquad \text{for~~} B\in\Borel(\R) \text{~with~}
	\sup B < \infty,
\end{equation}
as can be seen from bounding the spectral projection in terms of the semigroup,
$\1_{B}(x) \le \exp\{- (x- \sup B)\}$, $x\in\R$,
 and an explicit estimate using the Feynman-Kac representation, see e.g.\ \cite[Thm.~6.1]{MR1756112}.
From \eqref{fin-int} and \eqref{def:chi_L2} we get the following
estimate for the integral in the second line of \eqref{Helf-Sjo4}
\begin{align}
	\label{Helf-Sjo5}
  \int_0^{\infty}\dt & \tr\left(\sqrt V g_L^t(H_L)\sqrt V\right) \nonumber\\
  & = \int_{[-E_{0}-1,-L^{-b}]} \d \nu_L(x) \; \frac{\chi^-_L(-x)}{(-x)}
  		+ \int_{]-L^{-b},0]}  \d\nu_L(x)\; \frac{\chi^-_L(-x)}{(-x)} \nonumber\\
  & \le  L^{b} \nu_L\big([-E_{0}-1,-L^{-b}]\big) + c_0 L^b  \nu_L\big(]-L^{-b},0]\big) \nonumber \\
  & \le (1+c_{0}) \, \what{\nu}\big([-E_{0}-1,0]\big) \, L^{b} =: \Xi L^{b}.
\end{align}
Taken together, \eqref{Helf-Sjo3b}, \eqref{Helf-Sjo4} and \eqref{Helf-Sjo5} imply
\begin{multline}
	\label{Helf-Sjo6}
  \int_0^{L^a}\!\dt\, t  \, K_{L}^{(2)}(t)
   \le \Xi Q_{n}(L^{a-b})  L^{a+b}   \Big( C_{<} L^{b+n(-1+\varepsilon+b)} \\
   + C_{>} L^{d+b(n+1)} \,
    	 \e^{-c_{\textsc{ct}} L^{\varepsilon}/2}\Big)
\end{multline}
for every $L> L_{n}$. We recall that $0 <\varepsilon < 1-b$. Therefore we can choose $n$ large enough as to ensure
\begin{equation}
a+2b +n(-1+b+\varepsilon) < 0.
\end{equation}
Since we assumed $a<b$, we conclude that
\begin{equation}
	  \int_0^{L^a}\!\dt\, t  \, K_{L}^{(2)}(t) = \oh(1)
\end{equation}
as $L\to\infty$. The same holds true for $K_{L}^{(1)}$ by an analogous argument.
Thus, we have shown \eqref{kah}.
 \qed

%
\section{Proof of Lemma \ref{Lemma:Error2}}
\label{sec:proofError2}

We rewrite the difference in the integrand in \eqref{lemma2error} as
$\chi_{L}^{-}(x) [\chi_{L}^{+}(y) - 1_{[E,E_{0}]}(y) ] + 1_{[E,E_{0}]}(y) [ \chi_{L}^{-}(x) -
1_{[-E_{0},E]}(x)]$. Extending the $t$-integral in \eqref{lemma2error} up to $+\infty$,
the two error terms from removing the smoothing where $x\in [-E_{0}-1, -E_{0}]$ or $y\in[E_{0},E_{0}+1]$ are seen to be bounded from above by
\begin{equation}
 	\int_{\R^2} \frac{\d\mu_{\ac}(x,y)}{(y-x)^{2}} \;
		\Big[\1_{[-E_{0}-1,E]}(x)\1_{[E_{0},E_{0}+1]}(y)
			+ \1_{[-E_{0}-1,-E_{0}]}(x) \1_{[E,E_{0}+1]}(y) \Big].
\end{equation}
This $L$-independent expression is finite, because $y-x \ge 1$ in the compact support of the integrand
thanks to $E\in [-E_{0}+1,E_{0} -1]$. Thus, this error is of order $\Oh(1)$ as $L\to\infty$.

In order to estimate the two error terms in \eqref{lemma2error}
from removing the smoothing where $x\in [E- L^{-b},E]$ or $y\in[E,E+L^{-b}]$,
we use the inequality \eqref{eq:gamma} for the Lebesgue density of $\mu_{\ac}$ together with the elementary estimate $t \e^{-t\zeta} \le \zeta^{-1}$ for $t,\zeta > 0$. This gives the upper bound
\begin{multline}
  \label{eq:Error2:eq1}
  L^a 	\int_{\R^2}\!\d x\d y\; \frac{\gamma_{1}(x)\, \gamma_{2}(y)}{y-x} \;
  \Big[\1_{[-E_{0}-1,E]}(x)\1_{[E,E+L^{-b}]}(y)  \\
	+ \1_{[E-L^{-b},E]}(x) \1_{[E,E_{0}]}(y) \Big] .
\end{multline}
But this bound is of order $\Oh(L^{a-b}\ln L)$, as follows from the next lemma because $E$ is a Lebesgue point of both $\gamma_{1}$ and $\gamma_{2}$. Since $b>a$, the proof is complete.
\qed

\begin{lemma}
  Let $A >0$ and $\kappa_{1}, \kappa_{2} \in L^{1}_{\loc}(\R)$.
  Assume that $0$ is a Lebesgue point of both $\kappa_{1}$ and $\kappa_{2}$. Then we have
  \begin{equation}
  	\label{double-int}
    \int_{-A}^0\dx\int_0^{\eta} \!\dy\, \frac{\kappa_1(x)\,\kappa_2(y)}{y-x}
    = \Oh(\eta \ln\eta)
	\end{equation}
	as $\eta\downarrow 0$.
\end{lemma}
\begin{proof}
  W.l.o.g.\ we assume $\kappa_{1},\kappa_{2} \ge0$. Let $\eta \in {]0,1[}$.
  We start with an estimate for the $y$-integral for given $x<0$.
  The function $[0,\eta] \ni y \mapsto \int_0^y \!\d\zeta\, \kappa_2(\zeta)$ is absolutely
  continuous by the fundamental theorem of calculus for Lebesgue integrals.
  Therefore we can apply integration by parts to conclude
  \begin{align}
  	\label{y-int}
    \int_0^{\eta} \!\dy\, \frac{\kappa_2(y)}{y-x}
    &= \frac{\int_0^y \!\d\zeta\, \kappa_2(\zeta)}{y-x}
      \bigg|_{y=0}^{y=\eta}
      + \int_0^{\eta} \!\d y\, \frac{\int_0^y \d\zeta\,\kappa_2(\zeta)}{(y-x)^2}
       \nonumber\\
    &\le C_{2} \left( \frac{\eta}{\eta-x} +
      \int_0^{\eta} \!\dy\, \frac{y}{(y-x)^2} \right)
      = C_{2} \ln(1-\eta/x),
  \end{align}
  where $C_{2} := \sup_{y \in {]0,1[}} y^{-1} \int_0^y \d\zeta\,\kappa_2(\zeta)
  \in {]0,\infty[}$ is finite, because
  $0$ is a Lebesgue point of $\kappa_2$.

	The function $[-A,0] \ni x \mapsto -\int_x^0 \!\d\zeta\, \kappa_1(\zeta)$ is absolutely
  continuous so that we can use integration by parts to
  compute the $x$-integral of $\kappa_{1}$ times the bound
  \eqref{y-int}
  \begin{multline}
    \int_{-A}^{-\varepsilon}\!\dx\, \kappa_1(x) \, \ln(1- \eta/x)  \\
    = \bigg(-\int_x^0 \!\d\zeta\, \kappa_1(\zeta) \bigg) \ln(1- \eta/x)
      \bigg|_{x=-A}^{x = -\varepsilon}
    + \int_{-A}^{-\varepsilon} \!\d x\,  \frac{\eta/x^{2}}{1-\eta/x}
    \int_x^0 \d\zeta\,\kappa_1(\zeta),
 	\end{multline}
	where we introduced $\varepsilon \in {]0,A[}$ to exclude the singularity of the logarithm
	(and have dropped the constant $C_{2}$). Since $0$ is a Lebesgue point of $\kappa_1$, we
	infer the existence of 	a finite constant
	$C_{1} := \sup_{x \in {]-A,0[}} (-x)^{-1} \int_x^0 \d\zeta\,\kappa_1(\zeta) \in {]0,\infty[}$
	and can perform the limit $\varepsilon\downarrow 0$ to get
	\begin{align}
    \int_{-A}^0\dx\int_0^{\eta} \!\dy\, \frac{\kappa_1(x)\,\kappa_2(y)}{y-x}
    &\le C_{1}C_{2} \bigg[ A \ln(1+\eta/A) + \int_{-A}^{0}\!\d x\, \frac{\eta}{\eta -x}\bigg]
    \nonumber\\[.5ex]
    & = C_{1}C_{2} \big[ A \ln(1+\eta/A) + \eta \ln(1+ A/\eta) \big].
 	\end{align}
  This implies \eqref{double-int}.
\end{proof}

%
\appendix
\section*{Appendix. Relation to Scattering Theory}
\addtocounter{section}{1}
\setcounter{equation}{0}

In this section we conduct the proofs of Theorem~\ref{prop:scattering} and Corollary~\ref{corollary:anderson}. The goal is to express the decay exponent $\gamma(E)$ of the overlap in
Theorem~\ref{thm:main} in terms of quantities from Scattering Theory.

\begin{proof}[Proof of Theorem~\ref{prop:scattering}]
We recall that $V_{0}=0$ in this theorem. Hence, \eqref{acac} implies that
both $H$ and $H'$ have purely absolutely continuous spectrum with
$\spec(H) = \spec(H') = [0,\infty[$. Moreover, we know from
\eqref{gamma-def} and Proposition~\ref{prop:Birman} that
\begin{equation}\label{gamma-start}
  \gamma(E)
  = \lim_{\varepsilon\to 0}
  \varepsilon^{-2}
  \tr(\sqrt{V} P_E^\varepsilon V\varPi_E^\varepsilon\sqrt{V})
\end{equation}
exists for Lebesgue-a.e.\ $E\in\R$. This implies in particular that
$\gamma(E) = 0$ for $E < 0$.

The integral kernels of the spectral projections $P_E^{\varepsilon}$ and $\varPi_E^{\varepsilon}$
can be represented as
\begin{align}
	\label{kernel-rep}
  P_E^{\varepsilon}(x,y) & = \frac{1}{2(2\pi)^d}\int_0^\infty\!\d\lambda\,
    \lambda^{\frac{d-2}{2}} \1_{I_{E,\varepsilon}}(\lambda)
    \int_{\mathbb{S}^{d-1}}\!\d\varOmega(\omega) \,
    \e^{\i\sqrt{\lambda}y\cdot\omega} \,
    \e^{-\i\sqrt{\lambda}x\cdot\omega},
 	\nonumber\\[.5ex]
  \varPi_E^{\varepsilon}(x,y) & = \frac{1}{2(2\pi)^d}\int_0^\infty\!\d\lambda\,
    \lambda^{\frac{d-2}{2}}   \1_{I_{E,\varepsilon}}(\lambda)
    \int_{\mathbb{S}^{d-1}}\!\d\varOmega(\omega) \,
    \overline{\psi(y;\omega,\lambda)} \,\psi(x;\omega,\lambda)
\end{align}
for all $x,y\in\R^{d}$, where $I_{E,\varepsilon} := {]E -\varepsilon/2, E+\varepsilon/2[}$,
$x\cdot\omega$ stands for the Euclidean scalar product of $x$ with
$\omega$ (viewed as a unit vector in $\R^{d}$) and $\psi : \R^d\times\mathbb{S}^{d-1}\times[0,\infty[\to\C$ are the generalized
eigenfunctions of $H' = -\Delta + V$ due to Ikebe and Povzner, see Sect.~1.3 in \cite{MR1774673} and
references therein.
In particular, $\psi(\dotid;\omega,\lambda)$ is a solution
to the Lippmann-Schwinger equation.
Inserting \eqref{kernel-rep} in \eqref{gamma-start} gives
\begin{align}
	\label{gamma-a}
  \gamma(E) &= \frac{E^{d-2}}{4(2\pi)^{2d}}
  \int_{\mathbb{S}^{d-1}}\!\d\varOmega(\omega)  \int_{\mathbb{S}^{d-1}}\!\d\varOmega(\theta)
  \biggabs{\int_{\R^d}\!\dx\, V(x) \, \e^{\i\sqrt{E}x\cdot\theta}
  \psi(x;\omega,E)}^2 \nonumber\\
  &=
  \frac{E^{(d-1)/2}}{(2\pi)^{d+1}}
    \int_{\mathbb{S}^{d-1}}\!\d\varOmega(\omega)  \int_{\mathbb{S}^{d-1}}\!\d\varOmega(\theta)\,
    \abs{a(\theta,\omega; E)}^2
\end{align}
for Lebesgue-a.a.\ $E\ge 0$, where the scattering amplitude is defined by
\begin{equation}
  a(\theta,\omega;E) :=
  - \e^{\i\pi(d-3)/4}\frac{E^{(d-3)/4}}{2(2\pi)^{(d-1)/2}}
  \int_{\R^d}\!\dx\, \e^{-\i\sqrt{E}\,x \cdot\theta} \, V(x) \, \psi(x;\omega,E),
 \end{equation}
see \cite[Eq.\ (1.17)]{MR1774673}. Following
\cite[Sect.~8.5]{MR1774673}, the right-hand side of \eqref{gamma-a} can be rewritten
in terms of the total scattering cross-section
$\sigma(E,\omega) = \int_{\mathbb{S}^{d-1}}\d\varOmega(\theta)\,\abs{a(\theta,\omega;E)}^2$
or in terms of the scattering matrix $S(E)$, which proves the lemma.
\end{proof}

\begin{proof}[Proof of Corollary~\ref{corollary:anderson}]
The corollary is concerned with the special case $d = 3$ and $V$ spherically
symmetric. In particular, the total scattering cross-section
$\sigma(E) = \sigma(E,\omega)$ does not
depend on the incident direction $\omega \in \mathbb{S}^{2}$ for any $E \ge 0$.
Following \cite[Eq.\ (111)]{MR529429}, we can rewrite
$\sigma(E)$ in terms of the partial wave amplitudes $f_\ell(E)$ so that \eqref{gamma-s-mat} becomes
\begin{equation}
  \gamma(E)
  = \frac{E}{(2\pi)^4} \int_{\mathbb{S}_2}\!\d\varOmega(\omega)\, \sigma(E)
  = \frac{E}{\pi^2} \sum_{\ell=0}^\infty (2\ell+1)\abs{f_\ell(E)}^2
\end{equation}
for Lebesgue-a.a. $E \ge 0$. Now, \cite[Eq.\ (112b)]{MR529429} finishes the proof.
\end{proof}

\section*{Acknowledgement}
We thank Peter Otte for many stimulating discussions on this subject.


\newcommand{\etalchar}[1]{$^{#1}$}
\newcommand{\noopsort}[1]{}

\end{document}